\documentclass[10pt,twoside]{amsart}
\usepackage{amsmath, amsfonts, amsthm, amssymb}
\newtheorem{definition}{Definition}[section]
\newtheorem{lemma}[definition]{Lemma}
\newtheorem{theorem}[definition]{Theorem}

\numberwithin{equation}{section} 

\newcommand \Xbar {\overline X}
\newcommand \Ybar {\overline Y}
\newcommand \lambdab {\overline \lambda}
\newcommand \etab {\overline \eta}

\newcommand \rbar {\overline r}
\newcommand \sbar {\overline s}
\newcommand \be   {\begin{equation}}
\newcommand \ee   {\end{equation}}

\newcommand \N      {\mathbb{N}}
\newcommand \del \partial

\begin{document}
\bibliographystyle{plain}
\title[Plane-symmetric spacetimes with positive cosmological constant]
{Plane-symmetric spacetimes
\\
with positive cosmological constant. 
\\
The case of stiff fluids}
\author[Philippe G. L{\tiny e}Floch and Sophonie B. Tchapnda]
{
Philippe G. L{\scriptsize e}Floch$^1$ and Sophonie B. Tchapnda$^{2}$
}
\thanks{
$^1$ Laboratoire Jacques-Louis Lions
\& Centre National de la Recherche Scientifique,
Universit\'e Pierre et Marie Curie (Paris 6), 4 Place Jussieu, 75252 Paris,
France. 
\\
Blog: {\tt philippelefloch.org.} \, Email : {\tt contact@philippelefloch.org.} 
 \newline
$^2$ Mathematics Department, Faculty of Science, University of Yaounde I, POB 812, Yaounde, Cameroon.
Email: {\tt Tchapnda@uy1.uninet.cm.}
\newline
%
{\it First version:} November 2010. {\it Final version:} in May 2011. {\it Published in:}
 Adv. Theor. Math. Phys. 15 (2011), 1--27.
}

\date{}

\begin{abstract}  We consider plane-symmetric spacetimes satisfying Einstein's field equations
with positive cosmological constant, 
when the matter is a fluid whose pressure is equal to its mass-energy density
(i.e.~a so-called stiff fluid). We study the initial-value problem for the
associated Einstein equations and establish a global existence result. 
The late-time asymptotics of solutions is also rigorously derived, 
and we conclude that the spacetime approaches the de~Sitter spacetime
while the matter disperses asymptotically. A technical difficulty dealt with here
lies in the fact that solutions may contain vacuum states as well as velocities approaching the speed of light, 
both possibilities leading to singular behavior in the evolution equations. 
\end{abstract}
\maketitle 


\section{Introduction}\label{intro}

The study of global properties of cosmological spacetimes is
a fundamental problem in mathematical relativity, as it provides
 a first step toward understanding fundamental issues such as the
  structure of singularities and the cosmic censorship conjecture.
  Such a study can be reduced to investigating the global existence
   and asymptotic behavior of solutions to the Einstein equations,
   possibly coupled to the equations of motion for a specific matter model. 
In the present paper, we treat the class of perfect fluids whose
pressure $p$ and mass-energy density $\mu \geq 0$ coincide. This is a limiting case ($\gamma=2$)
with the class of pressure laws $p=(\gamma-1) \, \mu$, in which $\gamma \in [1,2]$ is referred
as the adiabatic exponent of the fluid. Our main result concerns 
the initial-value problem for the associated Einstein equations: 
we establish a global existence result and rigorously 
determine the late-time asymptotic behavior of solutions.
This allows us to conclude that the spacetime is future geodesically complete and 
approaches the de~Sitter spacetime whereas the matter asymptotically disperses.

Observe first that singularities generically arise in initially smooth solutions to the 
fluid equations, that is, 
shock waves in the general case $\gamma\in (1,2]$
and shell-crossing singularities in the case $\gamma=1$. 
This is true even when gravitational effects are taken into account \cite{RS}.
If the solution is to be continued beyond shock waves, it is necessary to lower the regularity 
of initial data and search for weak solutions, as investigated by LeFloch and
co-authors (cf.~the review \cite{LeFloch} and the references therein). 

On the other hand, existence of {\sl smooth} solutions even in a long-time evolution 
can sometime be established in physically interesting situations. This is especially true
when a cosmological constant is included, as we do in the present paper. 
Global-in-time solutions and the existence of future geodesically complete spacetimes 
can be established under a smallness condition on the initial data, 
as recognized by Tchapnda \cite{tchapnda} for $\gamma=1$ and  
under the assumption of plane symmetry and, later, without symmetry 
and for $\gamma\in(1,4/3)$, by Rodnianski and Speck \cite{RoS} and Speck \cite{speck,speck2}.

As far as the limiting case $\gamma=2$ is concerned,
plane symmetric spacetimes have been investigated by Tabensky and Taub \cite{TT}
and LeFloch and Stewart \cite{LS}. In particular, \cite{TT} relies on two different coordinate systems
 in their analysis, a comoving coordinate system in which the fluid
 is at rest, and a characteristic coordinate system.  On the other hand, the work \cite{LeFloch,LS}
introduced the notion of weakly regular solutions to the Einstein equations. 

In the present paper, we rely on areal coordinates, a coordinate system
in which the time is defined to be the area-radius function determined by
surfaces of symmetry. In these geometry-based coordinates, we prove a global-in-time existence
theorem (in the future direction) for plane-symmetric solutions to the Einstein-stiff fluid equations
with cosmological constant. Importantly, 
we also derive the leading asymptotic behavior of solutions and conclude with 
the future geodesic completeness of the constructed spacetime. 

Our analysis relies on a change of fluid variables that
allows us to write the fluid equations in a way analogous to the case of a massless
scalar field, and then to take advantage of techniques for semi-linear hyperbolic equations.
(A similar structure was observed in \cite{TNR}.)
 A specific technical difficulty overcome in this work 
originates in the fact that solutions may naturally contain 
vacuum states as well as velocities approaching the speed of light, both possibilities leading to singular
behavior in the 
evolution equations. 

Note finally that our results extend to compressible fluids 
the conclusions obtained by Tchapnda and Rendall \cite{TR} for 
the Vlasov equation of (collision-less) kinetic dynamics. 

The outline of the paper is as follows. Section $2$ is concerned with the derivation
of the field equations for stiff fluids under plane-symmetry. Next, in 
Section~3 we develop the local existence and uniqueness theory 
and then, in Section~4, determine the global geometry and asymptotic behavior of the spacetimes under consideration. 


\section{Einstein-stiff fluid equations}
\label{system}

\subsection*{Gravitational field equations}
\label{einstein eq} 

We consider spacetimes $(M,g)$ such that the manifold has the
topology $M=I\times \mathbb{T}^3$, where $I$ is a real interval and
$\mathbb{T}^3=S^1 \times S^1 \times S^1$ is the three-torus. The
metric $g$ and the matter fields are required to be invariant under
the action of the Euclidean group $E_2$ on the universal cover. It
is also required that the spacetime has an $E_2$-invariant Cauchy
surface of constant areal time. In such conditions the metric can be
expressed in the form
\begin{equation} \label{1.1}
  ds^2 = -e^{2\eta(t,x)}dt^2 + e^{2\lambda(t,x)}dx^2 + t^2
  (dy^2 + dz^{2}),
\end{equation}
where the time variable describes $t > 0$ and the spatial variable the interval $x \in [0,1]$, while
the variables $y$ and $z$ range in $[0,2\pi]$; the metric coefficients $\eta$ and $\lambda$ are periodic in $x$
with period $1$.
The Einstein equations read
\begin{align}\label{einstein}
G^{\alpha\beta} + \Lambda g^{\alpha\beta}  = 8 \pi
T^{\alpha\beta},
\end{align}
where $G^{\alpha\beta}$ is the Einstein tensor, $T^{\alpha\beta}$
 the energy-momentum tensor and $\Lambda$ is the cosmological constant
 which we assume to be positive.
We also introduce the notation
$$
\rho=e^{2\eta}T^{00}, \quad j=e^{\lambda+\eta}T^{01}, \quad S=e^{2\lambda}T^{11}, \quad
p=t^2T^{22}
$$
which defines the fluid variables of interest.

After a tedious computation in the above coordinates, (\ref{einstein}) take the form of the following
evolution and constraint equations
(where the subscripts $t, x$ denote partial differentiation):
\begin{equation}
\label{1.2}
e^{-2\eta} (2t\lambda_t+1) - \Lambda t^{2} = 8 \pi t^{2}\rho,
\end{equation}
\begin{equation}
\label{1.3}
e^{-2\eta} (2t\eta_t-1)+ \Lambda t^{2} = 8 \pi t^{2}S,
\end{equation}
\begin{equation}
\label{1.4}
\eta_x = -4 \pi t e^{\lambda+\eta}j,
\end{equation}
\begin{equation} \label{1.5}
e^{-2\lambda}\left(\eta_{xx} + \eta_x(\eta_x - \lambda_x)\right) -
e^{-2\eta}\left(\lambda_{tt}+(\lambda_t- \eta_t)(\lambda_t+\frac{1}{t})\right) + \Lambda
 = 8 \pi p.
\end{equation}

\subsection*{Stiff fluid equations}\label{fluid eq}

The so-called stiff fluid under consideration 
is an isentropic perfect fluid with energy density $\mu>0$ equal to its pressure, that is,  $p=\mu$.  
The $4$-velocity vector $U^\alpha$ of the fluid is normalized to be of unit length: $U^\alpha U_\alpha=-1$.
The plane symmetry allows us to set $U^\alpha:=\xi(e^{-\eta},e^{-\lambda}u,0,0)$,
where $\xi=(1-u^2)^{-1/2}$ is the relativistic factor and $u$ is the scalar velocity satisfying $|u|<1$.
The energy momentum tensor for the stiff fluid is
$$
T^{\alpha\beta} = \mu \, (2U^\alpha U^\beta+g^{\alpha\beta}),
$$
that is
\begin{equation}\label{fluid comp}
\aligned &T^{00}=e^{-2\eta}\frac{1+u^2}{1-u^2}\mu=:e^{-2\eta}\rho,
\qquad\quad
T^{01}=e^{-\lambda-\eta}\frac{2u\mu}{1-u^2}=:e^{-\lambda-\eta}j,
\\
&T^{11}=e^{-2\lambda}\frac{1+u^2}{1-u^2}\mu=:e^{-2\lambda}S,
\qquad\quad
 T^{22}=T^{33}=t^{-2}\mu,
\endaligned
\end{equation}
while, due to the above assumptions, all the other components vanish identically.

The stiff fluid equations read
\begin{equation}
\label{euler}
\nabla_\alpha T^{\alpha\beta}=0.
\end{equation}
We can assume that the components $T^{\alpha 2}$ and $T^{\alpha 3}$
vanish identically, while by computing the remaining two components
we arrive at the two evolution equations
\begin{equation}
\label{1.6}
\aligned
& \rho_t+e^{\eta-\lambda}j_x=-2\lambda_t\rho-2\eta_x e^{\eta-\lambda}j-\frac{2}{t}(\rho+\mu),
\\
& j_t+e^{\eta-\lambda}\rho_x=-2\lambda_tj-2\eta_x e^{\eta-\lambda}\rho-\frac{2}{t}j.
\endaligned
\end{equation}

The equations may be put into a simpler form, as follows. Observe
that the first-order principal part of (\ref{1.6}) is a strictly
hyperbolic system of two equations associated with the two distinct
speeds $\pm e^{\eta-\lambda}$. Introducing the Riemann invariants
\begin{equation}\label{riem}
r:=\frac{1+u}{1-u}\mu=\rho+j, \qquad  s:=\frac{1-u}{1+u}\mu=\rho-j,
\end{equation}
and the directional derivatives
$$
D^+:=\del_t+e^{\eta-\lambda}\del_x,
\qquad
D^-:=\del_t-e^{\eta-\lambda}\del_x,
$$
and then combining the equations in (\ref{1.6}) together, we obtain
\begin{equation}
\label{1.9}
\aligned
& D^+r=-2\left(\lambda_t+\eta_xe^{\eta-\lambda}+\frac{1}{t}\right)r-\frac{2}{t}\sqrt{rs},
\\
& D^-s=-2\left(\lambda_t-\eta_xe^{\eta-\lambda}+\frac{1}{t}\right)s-\frac{2}{t}\sqrt{rs}.
\endaligned
\end{equation}
Finally, the expressions for $\lambda_t$ and $\eta_x$ taken from (\ref{1.2}) and (\ref{1.4}) can be plugged in
(\ref{1.9}), and by setting $X=e^\eta \sqrt{r}$ and $Y=e^\eta \sqrt{s}$ we arrive at
\begin{equation}
\label{Euler}
\aligned
& D^+X = - \Lambda te^{2\eta}X-\frac{1}{t}Y,
\\
& D^-Y = - \frac{1}{t}X-\Lambda te^{2\eta}Y,
\endaligned
\end{equation}
which we will refer to as the {\sl stiff fluid equations} for the unknowns $r$ and $s$.


\subsection*{Basic properties}\label{properties}

It is easily checked that (\ref{1.5}) is a {\sl consequence} of the equations (\ref{1.2})--(\ref{1.4}), (\ref{1.6}). One can also check that (\ref{1.4}) is a constraint
equation, that is, it is automatically satisfied for all times once it is satisfied on an initial Cauchy hypersurface. Therefore, we will work with (\ref{1.2}), (\ref{1.3}) and (\ref{Euler}) for the unknowns $\eta$, $\lambda$, $r$ and $s$. Observe that by definition $r$ and $s$ must be non-negative. From the definition we see that $S=\rho=(r+s)/2$.

We will solve the initial-value problem with data prescribed on the hypersurface $t=1$. Observe that once the fluid variables have been determined, the metric coefficient $\eta$ is obtained
by integrating (\ref{1.3}) in the time direction, i.e.
\begin{equation}
\label{eta}
e^{-2 \eta(t,x)} = {e^{-2 \etab(x)} \over t} + {1 \over t} \int_1^t \tau^2 \big(\Lambda - 4 \pi (r+s)(\tau,x)\big) \, d\tau
\end{equation}
with $\etab:=\eta(1,\cdot)$. 
Next, $\eta$ being known, the following equation (obtained from (\ref{1.2}) and (\ref{1.3})),
\begin{equation}
\label{lambda_t}
\lambda_t(t,x) = \eta_t(t,x) + \Lambda te^{2\eta} -\frac{1}{t},
\end{equation}
is integrated in time to yield the second metric coefficient
\begin{equation}
\label{lambda}
\lambda(t,x) = \lambdab(1,x) + \int_1^t \lambda_t(\tau,x) \, d\tau
\end{equation}
with $\lambdab =\lambda(1,\cdot) $. 
Therefore, it will be enough to concentrate on the stiff fluid equations (\ref{Euler}) together with the metric equation (\ref{eta}), that determine an evolution system for the unknows $\eta$, $r$, $s$.

Observe that there exists some $T^*>1$ such that the right hand side term in (\ref{eta}) is positive on $[1,T^*)\times [0,1]$.
Estimates for $r$ and $s$ can easily be derived as follows. The expressions for $\lambda_t$ and $\eta_x$ taken from (\ref{1.2}) and (\ref{1.4}) can be plugged in (\ref{1.9}) to yield

\be
\label{rs}
\aligned
& D^+r = - \big( 8\pi t e^{2\eta}s+\Lambda t e^{2\eta}+{1 \over t} \big) \, r-{2\over t}\sqrt{rs},
\\
& D^-s = -\big( 8\pi t e^{2\eta}r+\Lambda t e^{2\eta}+{1 \over t} \big) \, s-{2\over t}\sqrt{rs}. 
\endaligned
\ee 
Using the fact that $r$ and $s$ are
positive, this implies
$$
D^+r\leq-t^{-1}r,  \qquad D^-s\leq-t^{-1}s, 
$$
and
integrating this along the characteristic curves associated with the
operators $D^\pm$ implies that \be \label{bound} r\leq r(1,\cdot)\,t^{-1},\ \ s\leq s(1,\cdot)\, t^{-1}. \ee
 As a consequence, if
$\eta$ is bounded then so are $X$ and $Y$.

A straightforward computation leads to the following result.

\begin{lemma}\label{X1Y1}
Set 
$$
\aligned
& b_1=(\lambda-\eta)_xe^{\eta-\lambda}-\Lambda te^{2\eta}, 
\qquad 
b_2=-2\Lambda t\eta_xe^{2\eta},
\\
& b_3=(\eta-\lambda)_xe^{\eta-\lambda}-\Lambda te^{2\eta}, 
\qquad
b=-\frac{1}{t}.
\endaligned
$$
If $X$ and $Y$ solve (\ref{Euler}) then $X_x$ and $Y_x$ satisfy
\begin{equation}
\label{Euler_x}
\aligned
& D^+X_x =  b_1X_x+bY_x+b_2X,
\\
& D^-Y_x = bX_x+b_3Y_x+b_2Y.
\endaligned
\end{equation}
\end{lemma}

The following result will be used to obtain bounds on derivatives of $X$ and $Y$.

\begin{lemma}\label{A(t)}
Set 
$$
\aligned
& K(t)=\sup\{(X+Y)(t,x)\ | \ x\in[0,1]\}, 
\\
& A(t)=\sup\{(|X_x|+|Y_x|)(t,x)\ | \ x\in[0,1]\},
\\
& v(t)=\sup\{|(\lambda-\eta)_x|e^{\eta-\lambda}+\Lambda te^{2\eta}+\frac{1}{t}\ | \ x\in[0,1]\}, 
\\
& h(t)=2\Lambda t\sup\{|\eta_x|e^{2\eta}\ | \ x\in[0,1]\}.
\endaligned
$$
If $(X,Y)$ and $(X_x,Y_x)$ solve (\ref{Euler}) and (\ref{Euler_x}), respectively, with $X_x(1)=(e^{\etab}\sqrt{\rbar})_x$ and $Y_x(1)=(e^{\etab}\sqrt{\sbar})_x$, then
\begin{equation}
\label{A-bound}
A(t)\leq A(1)+\int_1^t\big(v(\tau)A(\tau)+h(\tau)K(\tau)\big)\ d\tau.
\end{equation}
\end{lemma}

\begin{proof}
 Equations (\ref{Euler_x}) can be written in the form
$$
\aligned
& \frac{d}{dt} X_x(t,\gamma_1(t)) =  \big(b_1X_x+bY_x+b_2X\big)(t,\gamma_1(t)),
\\
& \frac{d}{dt} Y_x(t,\gamma_2(t)) = \big(bX_x+b_3Y_x+b_2Y\big)(t,\gamma_2(t)),
\endaligned
$$
where $\gamma_1$ and $\gamma_2$ are the integral curves corresponding to $D^+$ and $D^-$ respectively.\\
Integrating this over $[1,t]$, taking the absolute value in each equation, adding the resulting inequalities and taking the supremum of each term yields (\ref{A-bound}).
\end{proof}


\section{Local existence theory}
 
\subsection*{Main statement of this section} 

We are interested in regular solutions, defined as follows.

\begin{definition} A {\rm regular solution} to the plane-symmetric
Einstein-stiff fluid equations consists of two metric coefficients $\eta, \lambda$ and Riemann invariants $r,s$
given as continuously differentiable functions defined on $[1,T] \times [0,1]$ and periodic in space.
\end{definition}

We pose the initial-value problem by choosing some functions $\etab, \lambdab, \rbar, \sbar$ as periodic functions
on $[0,1]$ satisfying the constraint
\begin{equation}
\label{constraint} \etab_x = - 2 \pi t e^{\lambdab+\etab} \, (\rbar
- \sbar),
\end{equation}
and, on the initial hypersurface $t=1$, we impose
\begin{equation}
\label{initial}
(\eta, \lambda, r, s)(1,\cdot) = (\etab, \lambdab, \rbar, \sbar).
\end{equation}

\begin{theorem}[Local existence and uniqueness theory in the Riemann variables] 
\label{theo1}
Given periodic, continuously differentiable data $\etab, \lambdab, \rbar, \sbar$ prescribed
on the initial hypersurface $t=1$ and satisfying the constraint
(\ref{constraint}), there exists a future development which consists
of continuously differentiable functions  $\eta, \lambda, r, s$
defined on some time interval $[1,T)$ (with $T\in(1,\infty])$) that
are periodic in space and satisfy the stiff fluid equations
(\ref{Euler}), together with the evolution equations (\ref{1.2}) and
(\ref{1.3}). 
\end{theorem}

Once the Riemann invariants $r$ and $s$ are known, the primary fluid
variables $\mu$ and $u$ can be determined from equations (\ref{fluid
comp}) and (\ref{riem}):
$$\mu=\sqrt{rs}, \quad \qquad
u=\frac{\sqrt{r}-\sqrt{s}}{\sqrt{r}+\sqrt{s}}.
$$
By construction, the Riemann invariants are {\sl bounded,} and this property is equivalent to the following 
restriction in the fluid variables: 
\be
\label{fluid}
{1 \pm u \over 1 \mp u} \, \mu  \lesssim 1. 
\ee
Observe that this condition allows the density to vanish, and the velocity component $u$ to approach 
$\pm 1$, which is the normalized light-speed.  The condition is equivalent to saying
\be
\label{fluid2}
0 \leq \mu \lesssim 1 - |u|^2.
\ee

\begin{theorem}[Local existence and uniqueness theory in the fluid variables] 
\label{theo2}
Under the assumptions of Theorem~\ref{theo1}, the problem with initial data satisfying the uniform bound 
\eqref{fluid2} admits a local-in-time solution which is unique in the following (generalized) sense: 
if $\mu_1, u_1$ and $\mu_2, u_2$ denote fluid solutions to the same initial value problem, then 
$$
\aligned
& \text{either } \mu_1 = \mu_2 >0 \ \text{ and } \,  u_1= u_2, 
\\
&\text{ or } \mu_1 = \mu_2 =0 \, \text{ and } \, u_1, u_2 \, \text{ are arbitrary.}  
\endaligned
$$
\end{theorem} 
 

\subsection*{Proof of the local existence result}
 
We rely on an iterative argument and define a sequence $(\eta_n,
r_n, s_n)$ in the following way.
\begin{enumerate}

\item For $t\in[1,+\infty)$ and $x\in[0,1]$, we set
$(\eta_0, r_0, s_0)(t,x) := (\etab, \rbar, \sbar)(x)$, $T_0=+\infty$.

\item If $\eta_{n-1}$, $r_{n-1}$, $s_{n-1}$ are regular on
$[1,T_{n-1})\times [0,1]$ with $T_{n-1}\leq\infty$, then we define $T_n$ to be supremum of
all $t'\in(1,T_{n-1})$ such that 
$$
{e^{-2 \etab(x)} \over t} + {1 \over t} \int_1^t \tau^2 \left(\Lambda - 4 \pi (r_{n-1}+s_{n-1})(\tau,x)\right) \, d\tau >0
$$
for all $x\in[0,1]$ and $t\in[1,t']$,
and we then set  
\begin{equation}
\label{eta_n}
e^{-2 \eta_n(t,x)} = {e^{-2 \etab(x)} \over t} + {1 \over t} \int_1^t \tau^2 \big(\Lambda - 4 \pi (r_{n-1}+s_{n-1}\big)(\tau,x))\ d\tau.
\end{equation}

\item We define $r_n$ and $s_n$ such that $X_n=e^{\eta_n}\sqrt{r_n}$,
and $Y_n=e^{\eta_n}\sqrt{s_n}$ are solutions of the system
\begin{equation}
\label{Euler_n}
\aligned
& D^+_{n-1}X_n = a_{n-1}X_{n-1}+bY_{n-1},
\\
& D^-_{n-1}Y_n = bX_{n-1}+a_{n-1}Y_{n-1},
\endaligned
\end{equation}
where $a_{n-1}=- \Lambda te^{2\eta_{n-1}}$, $b=-{1 \over t}$. $D^{\pm}_{n}$ is the $D^{\pm }$-operator corresponding to the $n$-th iterate. We prescribe the same initial data (\ref{initial}) for all $n$.

\end{enumerate}

Observe that $T_n\geq T^*$ for all $n$, so that all the iterates are well-defined and regular on the fixed time interval $[1,T^*)$.

In order to prove that the sequence of iterates  converges to a
regular solution, we establish uniform bounds on the iterates as
well as their time and space derivatives, and we prove their uniform
convergence. This is done in a series of lemmas.

In the sequel we denote by $\| \ \|$ the sup-norm  on the
function space of interest, $C$ denotes a constant that may change
at each occurrence.

\begin{lemma}\label{bounds}
The sequences $\eta_n$, $X_n$, $Y_n$, $r_n$, $s_n$ and $(\eta_n)_t$
are  uniformly bounded in $n$, in the sup-norm by a continuous
function of $t$, on a time interval $[1,T^{(1)}]$.
\end{lemma}

\begin{proof}
Set 
$$
\aligned
& P_n(t):=\sup\{e^{2\eta_n(t,x)} \ | x\in[0,1]\},
\\
& K_n(t):=\sup\{(X_n+Y_n)(t,x)\ |  x\in[0,1]\}.
\endaligned
$$
Using equations (\ref{Euler_n}), we apply the same argument used in the proof of Lemma~\ref{A(t)} and obtain
\begin{equation}\label{Kn-bound}
K_n(t)\leq K_0+\int_1^t m_{n-1}(\tau)K_{n-1}(\tau)\ d\tau,
\end{equation}
with
\begin{align*}
m_n(t)&=\sup\{\Lambda t e^{2\eta_n}+\frac{1}{t} ; \ x\in[0,1]\}\\
&\leq t(1+\Lambda)(1+P_n(t)),
\end{align*}
so that
\begin{equation}\label{Kn-bound1}
K_n(t)\leq K_0+(1+\Lambda)\int_1^t \tau(1+ P_{n-1}(\tau))K_{n-1}(\tau)\ d\tau.
\end{equation}
On the other hand equation (\ref{eta_n}) implies
\begin{equation}
\label{eta n_t}
(\eta_n)_t=\frac{1}{2t}-\frac{\Lambda}{2} t e^{2\eta_n}+2\pi te^{2\eta_n-2\eta_{n-1}}(X_{n-1}^2+Y_{n-1}^2),
\end{equation}
and since $e^{-2\eta_{n-1}}\leq \frac{e^{-2\etab}+\Lambda t^3}{t}\leq C(1+\Lambda)t^2$, it follows that
\begin{equation}
\label{Pn-bound}
P_n(t)\leq \|e^{2\etab}\|+C(1+\Lambda)\int_1^t \tau^3(1+ K_{n-1}(\tau))^2(1+P_{n}(\tau))^2\ d\tau.
\end{equation}
Now defining $Q_n(t):=\sup\{K_m(t)+P_m(t) ; m\leq n\}$ and adding (\ref{Kn-bound}) and (\ref{Pn-bound}) we arrive at
\begin{equation}
\label{Qn-bound}
Q_n(t)\leq K_0+ \|e^{2\etab}\|+C(1+\Lambda)\int_1^t \tau^3(1+Q_{n}(\tau))^4\ d\tau.
\end{equation}
Let $[1,T^{(1)})$ (with $T^{(1)}\in(1,T^*]$) be the maximal interval of existence for 
the solution $z_1$ of the integral equation
$$
z_1(t)= K_0+ \|e^{2\etab}\|+C(1+\Lambda)\int_1^t \tau^3(1+z_1(\tau))^4\ d\tau, \ \ z_1(1)= K_0+ \|e^{2\etab}\|.
$$
Then $Q_n(t)\leq z_1(t)$, for all $n\in \N$ and $t\in (1,T^{(1)})$. The same is true for $K_n$ and $P_n$. It follows that $\eta_n$, $X_n$, $Y_n$, and then $r_n$, $s_n$ and $(\eta_n)_t$ are uniformly bounded. To bound $(\eta_n)_t$, we use (\ref{eta n_t}).
\end{proof}

\begin{lemma}\label{der-bounds}
The sequences $(\eta_n)_x$, $(X_n)_x$, $(Y_n)_x$,  $(X_n)_t$,
$(Y_n)_t$, $(r_n)_x$, $(s_n)_x$, $(r_n)_t$ and $(s_n)_t$ are
uniformly bounded in $n$, the sup-norm by a continuous function of
$t$ on a time interval $[1,T^{(2)}]$.
\end{lemma}

\begin{proof} Set
$$
\aligned
A_n(t) :=& \sup\{|(X_n)_x|+|(Y_n)_x|(t,x)\ |  x\in[0,1]\},
\\
A_0:=& \sup\{(|\Xbar_x|+|\Ybar_x|)(x)\ | x\in[0,1]\},
\\
B_n(t):=&\sup\{|(e^{-2\eta_n(t,x)})_x| \ | x\in[0,1]\}.
\endaligned
$$
Then taking the spatial derivative in (\ref{Euler_n}) gives the following equations:
\begin{align*}
D_{n-1}^+(X_n)_x =&  (\lambda_{n-1}-\eta_{n-1})_xe^{\eta_{n-1}-\lambda_{n-1}}(X_n)_x-2\Lambda t(\eta_{n-1})_xe^{2\eta_{n-1}}X_n\\&-\Lambda te^{2\eta_{n-1}}(X_{n-1})_x-\frac{1}{t}(Y_{n-1})_x ,
\end{align*}
\begin{align*}
D_{n-1}^-(Y_n)_x =&  (\eta_{n-1}-\lambda_{n-1})_xe^{\eta_{n-1}-\lambda_{n-1}}(Y_n)_x-2\Lambda t(\eta_{n-1})_xe^{2\eta_{n-1}}Y_n\\&-\Lambda te^{2\eta_{n-1}}(Y_{n-1})_x-\frac{1}{t}(X_{n-1})_x.
\end{align*}
But, using Lemma~\ref{bounds}, we have 
\begin{align*}
|(\lambda_{n-1}-\eta_{n-1})_x(s)|&=|(\lambdab-\etab)_x+2\Lambda\int_1^s \tau(\eta_{n-1})_xe^{2\eta_{n-1}}\ d\tau|\\&\leq Cs^2(1+B_{n-1}(s)),
\end{align*}
so that applying the same argument as in Lemma~\ref{A(t)} and using Lemma~\ref{bounds} again, we obtain
\begin{equation}
\label{An-bound}
A_n(t)\leq A_0+C \int_1^t \tau^2 (1+B_{n-1}(\tau))(1+A_{n-1}(\tau)+A_{n}(\tau)))\ d\tau.
\end{equation}
On the other hand we have
\begin{equation}
\label{eta_n x}
(e^{-2 \eta_n(t,x)})_x = {-2\etab_x e^{-2 \etab} \over t} - {4 \pi \over t} \int_1^t \tau^2 (r_{n-1}+s_{n-1})_x(\tau,x))\ d\tau,
\end{equation}
which implies
\begin{equation}
\label{Bn-bound}
B_n(t)\leq 2\|\etab_x e^{-2 \etab}\|+C \int_1^t \tau^2 (A_{n-1}+B_{n-1})(\tau)\ d\tau.
\end{equation}
We have used the fact that
\begin{align*}
|(r_n+s_n)_x|&= |(e^{-2 \eta_n})_x(X_n^2+Y_n^2)+2e^{-2\eta_{n}}\big(X_n(X_n)_x+Y_n(Y_n)_x\big)|\\
&\leq C(A_{n}+B_{n})(t).
\end{align*}
Now defining $E_n(t):=\sup\{A_m(t)+B_m(t) ; m\leq n\}$ and adding (\ref{An-bound}) and (\ref{Bn-bound}) we arrive at
\begin{equation}
\label{En-bound}
E_n(t)\leq A_0+ 2\|\etab_x e^{-2 \etab}\|+C\int_1^t \tau^2(1+E_{n}(\tau))^2\ d\tau.
\end{equation}
Let $[1,T^{(2)})$ (with $T^{(2)}\leq T^{(1)}$) be the maximal interval of existence for the solution $z_2$ of the integral equation
$$
\aligned
z_2(t) & = A_0+ 2\|\etab_x e^{-2 \etab}\|+C\int_1^t \tau^2(1+z_{2}(\tau))^2\ d\tau, 
\\
z_2(1) & = A_0+ 2\|\etab_x e^{-2 \etab}\|.
\endaligned 
$$
Then $E_n(t)\leq z_2(t)$, for all $n\in \N$ and $t\in (1,T^{(2)})$. The same is true for $A_n$ and $B_n$. It follows that $(\eta_n)_x$, $(X_n)_x$, $(Y_n)_x$, $(X_n)_t$, $(Y_n)_t$, $(r_n)_x$, $(s_n)_x$ and then $(r_n)_t$, $(s_n)_t$ and $(\eta_n)_{tx}$ are uniformly bounded.
\end{proof} 

\begin{lemma}\label{conv}
The sequences $(\eta_n)$, $(X_n)$, and $(Y_n)$ converge uniformly on
$[1,T^{(3)}]$ for all $T^{(3)}$ less than $T^{(2)}$.
\end{lemma}

\begin{proof}
For $t\in[1,T^{(3)}]$, define 
$$
\aligned
\theta_n(t):=& \sup\{|X_{n+1}-X_n|(t,x)+|Y_{n+1}-Y_n|(t,x); x\in[0,1]\},
\\
\alpha_n(t):=& \sup\{\|(\eta_{n+1}-\eta_n)(s)\|+\|(X_{n+1}-X_n)(s)\|+\|(Y_{n+1}-Y_n)(s)\|; s\in[1,t]\},
\\
\Tilde{X}_n:=&X_{n+1}-X_n, \qquad \Tilde{Y}_n:=Y_{n+1}-Y_n.
\endaligned
$$
Combining equations (\ref{Euler_n}) written for $n+1$ and $n$ gives
\begin{equation}
\label{tilde X}
\aligned
& D^+_{n}\Tilde X_n = a_{n}\Tilde X_{n-1}+b\Tilde Y_{n-1}+F_n,
\\
& D^-_{n}\Tilde Y_n = b\Tilde X_{n-1}+a_{n}\Tilde Y_{n-1}+G_n,
\endaligned
\end{equation}
with 
$$
\aligned
F_n &=-(e^{2\eta_n}-e^{2\eta_{n-1}})\Lambda tX_{n-1}-(e^{\eta_n-\lambda_n}-e^{\eta_{n-1}-\lambda_{n-1}})(X_n)_x, 
\\
G_n & =-(e^{2\eta_n}-e^{2\eta_{n-1}})\Lambda tY_{n-1}+(e^{\eta_n-\lambda_n}-e^{\eta_{n-1}-\lambda_{n-1}})(Y_n)_x.
\endaligned
$$
Reasoning as in the proof of Lemma~\ref{A(t)} we have 
\begin{equation}
\label{theta-bound}
\theta(t)\leq \int_1^t\big(m_n(\tau)\theta_{n-1}+\sup\{|F_n(\tau,x)|+|G_n(\tau,x)| ; x\in[0,1]\}\big)\ d\tau,
\end{equation}
and this implies that
\begin{equation}
\label{tilde X n-bound}
|\Tilde X_n|+|\Tilde Y_n|\leq C \int_1^t\alpha_{n-1}(\tau)\ d\tau,
\end{equation}
we have used the mean value theorem to handle the terms $e^{2\eta_n}-e^{2\eta_{n-1}}$ and\\ $e^{\eta_n-\lambda_n}-e^{\eta_{n-1}-\lambda_{n-1}}$, and the previous lemmas.\\ On the other hand equation (\ref{eta n_t}) implies 
\begin{align*} 
(\eta_{n+1}-\eta_n)_t
= & -\frac{\Lambda}{2} t (e^{2\eta_{n+1}}-e^{2\eta_n}) 
  +2\pi te^{2\eta_{n+1}-2\eta_{n}}\big((X_{n+1}^2-X_{n}^2)+(Y_{n+1}^2-Y_{n}^2)\big)\notag\\
&+2\pi t(e^{2\eta_{n+1}-2\eta_{n}}-e^{2\eta_{n}-2\eta_{n-1}})(X_n^2+Y_n^2),
\end{align*}
and using Lemma \ref{bounds} and the mean value theorem it follows after integration in time that
$$
|\eta_{n+1}-\eta_n|\leq C \int_1^t(|\eta_{n+1}-\eta_n|+|\eta_{n}-\eta_{n-1}|+|X_{n+1}-X_n|+|Y_{n+1}-Y_n|)(\tau)\ d\tau,
$$
so that
\begin{equation}
\label{tilde eta n-bound}
|\eta_{n+1}-\eta_n|\leq C \int_1^t(\alpha_n+\alpha_{n-1})(\tau)\ d\tau.
\end{equation}
Combining (\ref{tilde X n-bound}) and (\ref{tilde eta n-bound}) leads to
$$
\alpha_n(t)\leq C \int_1^t(\alpha_n+\alpha_{n-1})(\tau)\ d\tau,
$$
which, by Gronwall's inequality, implies
$$
\alpha_n(t)\leq C \int_1^t\alpha_{n-1}(\tau)\ d\tau,
$$
and by induction
$$
\alpha_n(t)\leq \frac{C^{n+1}}{n!},
$$
and so $\alpha_n\to 0$ as $n\to\infty$. This establishes the uniform convergence of $\eta_n$, $X_n$, and $Y_n$.
\end{proof}

It follows from (\ref{eta n_t}) that the sequence $(\eta_n)_t$ converges uniformly as well. In the following lemma, the uniform convergence of other iterates derivatives is proven.

\begin{lemma}\label{der-conv}
The sequences $(\eta_n)_x$, $(X_n)_x$, $(Y_n)_x$, $(X_n)_t$ and $(Y_n)_t$ converge uniformly on $[1,T^{(4)}]$, where $[1,T^{(4)}] \subset [1,T^{(3)}]$.
\end{lemma}

\begin{proof} We set 
$$
\beta_n(t):=\sup\Big\{\|(\eta_{n+1}-\eta_n)_x(s)\|+\|(X_{n+1}-X_n)_x(s)\|+\|(Y_{n+1}-Y_n)_x(s)\|; s\in[1,t]\Big\}.
$$
Taking the space derivative in equations (\ref{Euler_n}) gives
\begin{equation}\label{C n}
D_n^+(X_{n+1})_x=\Tilde C_n, \qquad
 D_n^-(Y_{n+1})_x=\Tilde D_n
\end{equation}
 with
\begin{align*}
\Tilde C_n&=(\lambda_n-\eta_n)_x e^{\eta_n-\lambda_n}(X_{n+1})_x-2\Lambda t(\eta_n)_xe^{2\eta_n} X_n-\Lambda te^{2\eta_n}(X_n)_x-\frac{1}{t}(Y_n)_x \\
\Tilde D_n&=(\eta_n-\lambda_n)_x e^{\eta_n-\lambda_n}(Y_{n+1})_x-2\Lambda t(\eta_n)_xe^{2\eta_n} Y_n-\Lambda te^{2\eta_n}(Y_n)_x-\frac{1}{t}(X_n)_x.
\end{align*}
Let $\gamma_n^1$ and $\gamma_n^2$ be 
the integral curves corresponding to $D_n^+$ and $D_n^-$ respectively, that start from the point $(s,x)$ that is, for each n,
\begin{equation}\label{gamma}
(\gamma_n^1)_t=e^{\eta_n-\lambda_n}, \ (\gamma_n^2)_t=-e^{\eta_n-\lambda_n}, \ \gamma_n^1(s)=\gamma_n^2(s)=x.
\end{equation}
Integrating the first equation in (\ref{C n}) along $\gamma_n^1$, the second one along $\gamma_n^2$ yields after subtraction
\begin{equation}
\label{X n+1-X n}
\aligned
& (X_{n+1}-X_n)(s)=\int_1^s\big(\Tilde C_{n}(\tau,\gamma_n^1(\tau))-\Tilde C_{n-1}(\tau,\gamma_{n-1}^1(\tau))\big)\ d\tau,
\\
& (Y_{n+1}-Y_n)(s)=\int_1^s\big(\Tilde D_{n}(\tau,\gamma_n^1(\tau))-\Tilde D_{n-1}(\tau,\gamma_{n-1}^1(\tau))\big)\ d\tau.
\endaligned
\end{equation}
But we have 
\begin{equation}\label{tilde (Cn-Cn-1)}
\aligned
& |\Tilde C_{n}(\tau,\gamma_n^1(\tau))-\Tilde C_{n-1}(\tau,\gamma_{n-1}^1(\tau)|
\\
& \leq |\Tilde C_{n}(\tau,\gamma_n^1(\tau))-\Tilde C_{n}(\tau,\gamma_{n-1}^1(\tau)|+|(\Tilde C_{n}-\Tilde C_{n-1})(\tau)|.
\endaligned
\end{equation}

Given now any $\varepsilon>0$, we find, for any sufficiently large $n$,
\begin{equation}\label{tilde Cn}
|\Tilde C_{n}(\tau,\gamma_n^1(\tau))-\Tilde C_{n}(\tau,\gamma_{n-1}^1(\tau)|\leq C\varepsilon,
\end{equation}
we have used the uniform convergence of $\eta_n$, the uniform continuity of $\Tilde C_n$ over the compact set $[1,T^{(4)}]\times\big(\gamma_n^1([1,T^{(4)}])\cup\gamma_{n-1}^1([1,T^{(4)}])\big)$, and the following inequality which follows from (\ref{gamma})
\begin{equation}\label{gamma n}
|\gamma_n^1-\gamma_{n-1}^1|(\tau)
\leq C \, \sup \Big\{ \|(e^{2\eta_n}-e^{2\eta_n-1})(t)\| \ ; \ t\in[1,T^{(4)}] \Big\}.
\end{equation}
For the second term of the right hand side in (\ref{tilde (Cn-Cn-1)}) we have 
\begin{eqnarray*}
\aligned
&\Tilde C_{n}-\Tilde C_{n-1} 
\\
& =\big((\lambda_{n}-\eta_{n})_x-(\lambda_{n-1}-\eta_{n-1})_x\big)e^{\eta_n-\lambda_n}(X_{n+1})_x\\
& \quad+
(\lambda_{n-1}-\eta_{n-1})_x\Big(e^{\eta_n-\lambda_n}(X_{n+1}-X_n)_x+(e^{\eta_n-\lambda_n}-e^{\eta_{n-1}-
\lambda_{n-1}})(X_{n+1})_x\Big)\\
& \quad 
-2\Lambda t (\eta_{n}-\eta_{n-1})_xe^{2\eta_n}X_n-2\Lambda t(\eta_{n-1})_x \Big(e^{2\eta_n}(X_{n}-X_{n-1})+(e^{2\eta_n}-e^{2\eta_{n-1}})(X_{n-1})\Big)\\
& \quad -\Lambda t (e^{2\eta_{n}}-e^{2\eta_{n-1}})(X_n)_x-\Lambda te^{2\eta_{n-1}}\big((X_{n})_x-(X_{n-1})_x\big)-\frac{1}{t}\big((Y_{n})_x-(Y_{n-1})_x\big),
\endaligned
\end{eqnarray*}
and 
$$
(\lambda_{n}-\eta_{n})_x=(\lambdab-\etab)_x+2\Lambda\int_1^t t(\eta_n)_xe^{2\eta_n}\ d\tau,
$$
 so that
\begin{align*}
|(\lambda_{n}-\eta_{n})_x-(\lambda_{n-1}-\eta_{n-1})_x|\leq C\varepsilon +C\sup\{\|(\eta_n-\eta_{n-1})_x(t)\| \ ; t\in[1,T^{(4)}]\}.
\end{align*}
Thus, for $n$ sufficiently large,
\begin{equation}\label{bound tilde Cn}
 \|(\Tilde C_{n}-\Tilde C_{n-1})(\tau)\|\leq C\varepsilon+C(\beta_n+\beta_{n-1})(\tau).
\end{equation}
It then follows from (\ref{X n+1-X n})-(\ref{tilde Cn}) and (\ref{bound tilde Cn}) that for $n$ sufficiently large,
\begin{equation}
\label{Y n+1-Y n}
\aligned
& |(X_{n+1}-X_n)_x|(s)\leq C\varepsilon+C\int_1^s(\beta_n+\beta_{n-1})(\tau)\ d\tau,
\\
& |(Y_{n+1}-Y_n)_x|(s)\leq C\varepsilon+C\int_1^s(\beta_n+\beta_{n-1})(\tau)\ d\tau.
\endaligned
\end{equation}
On the other hand, taking the spatial derivative in (\ref{eta n_t}), subtracting the resulting equations written for $n+1$ and $n$ gives
\begin{eqnarray*}
\aligned
& \hskip-.3cm (\eta_{n+1}-\eta_{n})_{tx} 
\\
= & -\Lambda t (\eta_{n+1}-\eta_{n})_xe^{2\eta_{n+1}}-\Lambda t(\eta_{n})_x (e^{2\eta_{n+1}}-e^{2\eta_{n}})\\&+4\pi t(\eta_{n+1}-\eta_{n})_xe^{2(\eta_{n+1}-\eta_n)}(X_{n}^2+Y_{n}^2)
-4\pi t(\eta_{n}-\eta_{n-1})_xe^{2(\eta_{n}-\eta_{n-1})}(X_{n-1}^2+Y_{n-1}^2)\\&
+4\pi te^{2(\eta_{n+1}-\eta_n)}\big((X_{n}-X_{n-1})_xX_n+(Y_{n}-Y_{n-1})_xY_n\big)\\&
+4\pi te^{2(\eta_{n+1}-\eta_n)}\big((X_{n-1})_xX_n+(Y_{n-1})_xY_n\big)\\&
-4\pi te^{2(\eta_{n}-\eta_{n-1})}\big((X_{n-1})_xX_{n-1}+(Y_{n-1})_xY_{n-1}\big),
\endaligned
\end{eqnarray*}
from this and the previous lemmas, it follows that, for $n$ sufficiently large,
\begin{equation}
\label{eta n+1-eta n}
|(\eta_{n+1}-\eta_n)_x|(s)\leq C\varepsilon+C\int_1^s(\beta_n+\beta_{n-1})(\tau)\ d\tau.
\end{equation}
Combining (\ref{Y n+1-Y n}) and (\ref{eta n+1-eta n}), and taking the supremum over $s\in[1,t]$ yields, for $n$ sufficiently large,
$$ 
\beta_n(t)\leq C\varepsilon+C\int_1^t(\beta_n+\beta_{n-1})(\tau)\ d\tau,
$$
and by Gronwall's lemma it follows that, for $n$ sufficiently large and $t\in[1,T^{(4)}]$,
$$
\delta_n(t)\leq C\varepsilon,
$$
where $\delta_n(t):=\sup\{\beta_m, m\leq n\}$.
The uniform convergence of $(\eta_n)_x$, $(X_n)_x$, $(Y_n)_x$, $(X_n)_t$ and $(Y_n)_t$ follows.
\end{proof}

Lemmas \ref{conv} and \ref{der-conv} allow us to pass to the limit in (\ref{eta_n}) and (\ref{Euler_n}) and obtain a regular solution $(\eta, X, Y)$ to our system on a time interval $[1,T)$.
It is easily checked that this solution is unique. Namely, 
let $(\eta_i, X_i, Y_i)$, $i=1, 2$, be two regular solutions of the Cauchy problem for the same initial data $(\etab,\Xbar, \Ybar)$ at $t=1$. Using the same argument as in the proof of iterates convergence leads to
$$
 \alpha(t)\leq C\int_1^t\alpha(\tau)\ d\tau,
$$
where $\alpha(t)=\sup\{\|(\eta_1-\eta_2)(s)\|+\|(X_1-X_2)(s)\|+\|(Y_1-Y_2)(s)\| \ ; s\in[1,t]\}$.\\ It follows that $\alpha(t)=0$, for $t\in[1,T)$ i.e. the solution is unique.

We have thus established the existence of a unique, local-in-time 
regular solution $(\eta, \lambda, r, s)$ to the Cauchy problem for the plane symmetric Einstein-stiff fluid equations written in areal coordinates.


\section{Global existence theory and asymptotics}

\subsection*{Global existence}

We are now in a position to establish the following main result, which takes advantage of our assumption $\Lambda>0$. 

\begin{theorem}[Global existence theory and asymptotics]
Under the assumptions in Theorem~\ref{theo1}, 
the solution constructed therein
is defined up to $T=+\infty$, the spacetime is future geodesically
complete, and the following asymptotic properties hold at late times:
\begin{equation}
\aligned &\eta=-\ln t(1+O\big((\ln t)^{-1})\big),\ \ \lambda=\ln
t(1+O\big((\ln t)^{-1})\big), 
\\
&  r=O(t^{-1}),\qquad s=O(t^{-1}), 
\\
&\eta_t=-\frac{1}{t}(1+O(t^{-1})),\ \
\lambda_t=\frac{1}{t}(1+O(t^{-1})),
\\
& \eta_x=O(1). 
\endaligned
\end{equation}
Consequently, the generalized Kasner exponents associated with this spacetime 
(cf.~\eqref{Kasner-expo}, below) 
tend to $1/3$:
$$
\lim_{t\to\infty}\frac{\kappa^1_1(t,x)}{\kappa(t,x)}=\lim_{t\to\infty}\frac{\kappa^2_2(t,x)}{\kappa(t,x)}=\lim_{t\to\infty}\frac{\kappa^3_3(t,x)}{\kappa(t,x)}=\frac{1}{3},
$$ 
where $\kappa=\kappa^i_i$ denotes the trace of the second fundamental form $\kappa_i^j$.
\end{theorem}

In particular, this shows that the spacetime approaches the {\sl de Sitter spacetime} 
asymptotically. To establish this global result, we begin with a continuation criterion, based on 
the same notation as in the previous section.

\begin{lemma}\label{criterion}
Let $[1,T)$ be the maximal interval of existence of solutions to the system under consideration. 
 If $\sup\{|\eta(t,x)| \ | x\in[0,1], t\in[1,T)\}<+\infty$ then $T=+\infty$.
\end{lemma}

\begin{proof}
It suffices to prove that under the assumption that $\eta$ is bounded on $[1,T)$, the same is true for $\eta_x$, $\eta_t$, $X$, $Y$, $X_x$, $Y_x$, $X_t$, and $Y_t$. 
First of all, by definition we have $X=e^\eta\sqrt{r}$ and $Y=e^\eta\sqrt{s}$
and it follows from the decay inequalities (\ref{bound}) that $X$ and $Y$ are bounded.
Next, recalling that $\eta_x=-2\pi te^{\eta+\lambda}(r-s)$ and $\eta_t=\frac{1}{2t}+2\pi t(X^2+Y^2)-\frac{\Lambda t}{2}e^{2\eta}$, we find that
 $\eta_x$ and $\eta_t$ are bounded as well. Here, we have used the fact that
$$
\lambda (t,x)=(\lambdab-\etab)(x)+\eta(t,x)-\ln t+\Lambda\int_1^t\tau e^{2\eta}(\tau,x)\ d\tau.
$$
Taking the spatial derivative in this equation implies
$$
(\lambda -\eta)_x(t,x)=(\lambdab-\etab)(x)+2\Lambda\int_1^t\tau (\eta_x e^{2\eta})(\tau,x)\ d\tau,
$$
so that $v(t)$, defined in Lemma~\ref{A(t)}, is bounded. Rewriting (\ref{A-bound})
$$
A(t)\leq A(1)+\int_1^t\big(v(\tau)A(\tau)+h(\tau)K(\tau)\big)\ d\tau,
$$
and using the fact that $h$ and $K$ are bounded, Gronwall's lemma
allows us to conclude that $A$, and then $X_x$ and $Y_x$ are
bounded. Bounds on $X_t$ and $Y_t$ then follow from (\ref{Euler}).
\end{proof}

We now prove that $\eta$ is bounded in order to conclude that $T=+\infty$.

\begin{lemma}\label{global}
The function $\eta$ satisfies
$$
\sup\big\{|\eta(t,x)| \ \, / \,  x\in[0,1], t\in[1,T) \big\}<+\infty.
$$
\end{lemma}

\begin{proof}
We can deduce from (\ref{eta}) that $e^{-2\eta(t,x)}\leq \frac{e^{-2\etab(x)}+\Lambda t^3}{t}$, i.e.
\begin{equation}\label{e 2eta}
e^{2\eta(t,x)}\geq \frac{t}{C+\Lambda t^3},
\end{equation}
which provides a (negative, say) lower bound on $\eta$. 
Now, let us prove that
\begin{equation}\label{e eta+lambda}
\int_0^1(e^{\eta+\lambda}\rho)(t,x)\ dx\leq Ct^{-4}, \ t\in[1,T), \ x\in[0,1],
\end{equation}
which will eventually lead us to an upper bound for  $\eta$. 

Using the equations (\ref{1.2}), (\ref{1.3}), and (\ref{1.6}), 
 after some computations we find  
\begin{align*}
\frac{d}{dt}\Bigg(\int_0^1(e^{\eta+\lambda}\rho)(t,x)\ dx\Bigg)
=&\int_0^1e^{\eta+\lambda}\rho\big(-\frac{1}{t}-\Lambda t e^{2\eta}\big)\ dx-\int_0^1e^{2\eta}\big(j_x+2\eta_x j\big)\ dx\\&-\int_0^1\frac{2}{t}e^{\eta+\lambda}\mu\ dx.
\end{align*}
Since $\mu \geq 0$ and 
$$
\int_0^1e^{2\eta}\big(j_x+2\eta_x j\big)\ dx=\int_0^1\big(e^{2\eta} j\big)_x\ dx=0,
$$
it follows that
\begin{equation}\label{rho}
\frac{d}{dt}\Bigg(\int_0^1(e^{\eta+\lambda}\rho)(t,x)\ dx\Bigg)
\leq
\frac{1}{t}\int_0^1e^{\eta+\lambda}\rho\big(-1-\Lambda t^2 e^{2\eta}\big)\ dx.
\end{equation}
Thanks to (\ref{e 2eta}), we have
 $$-\Lambda t^2e^{2\eta}\leq\frac{-\Lambda t^3}{C+\frac{\Lambda}{3}t^3}\leq-3+\frac{9C}{\Lambda}t^{-3},$$ so that (\ref{rho}) implies
\begin{align*}
\frac{d}{dt}
\Bigg( t^4\int_0^1(e^{\eta+\lambda}\rho)(t,x)\ dx\Bigg)
& =4t^3\int_0^1(e^{\eta+\lambda}\rho)(t,x)\ dx+t^4\frac{d}{dt}\Bigg(\int_0^1(e^{\eta+\lambda}\rho)(t,x)\ dx\Bigg)
\\
&\leq4t^3\int_0^1(e^{\eta+\lambda}\rho)(t,x)\ dx+t^3
\int_0^1e^{\eta+\lambda}\rho\big(-4+\frac{9C}{\Lambda t^3}\big)\ dx\\
&\leq\frac{9C}{\Lambda t^4}\, t^4\int_0^1(e^{\eta+\lambda}\rho)(t,x)\ dx,
\end{align*}
from which we deduce (\ref{e eta+lambda}) by integration.

We are now in a position to make use of the integral estimate (\ref{e eta+lambda}).
Recalling that $\eta_x=-4\pi te^{\eta+\lambda}j$ and $0\leq j\leq \rho$, we control 
the spatial oscillation of $\eta$ at each time, as follows: 
\begin{align*}
\Bigg| 
\eta(t,x)-\int_0^1\eta(t,\tau)\ d\tau \Bigg|
&=\Bigg| \int_0^1\int_\tau^x\eta_x(t,\sigma)\ d\sigma\ d\tau \Bigg|
\leq
\int_0^1\int_0^1|\eta_x(t,\sigma)|\ d\sigma\ d\tau\\&
\leq 4\pi t\int_0^1(e^{\eta+\lambda}j)(t,\sigma)\ d\sigma
\leq4\pi t\int_0^1(e^{\eta+\lambda}\rho)(t,\sigma)\ d\sigma,
\end{align*}
and thanks to (\ref{e eta+lambda}), this implies that
\begin{equation}\label{eta-int}
\Big|
\eta(t,x)-\int_0^1\eta(t,\tau) \, d\tau \Big|
\leq Ct^{-3}, \qquad t\in[1,T), \ x\in[0,1].
\end{equation}

We will have the desired upper bound on $\eta$, provided we can control its integral. 
Recalling that $\eta_t-\lambda_t=\frac{1}{t}-\Lambda t e^{2\eta}$ and using (\ref{e 2eta}) gives
$$
\frac{\partial}{\partial t}e^{\eta-\lambda}=(\eta_t-\lambda_t)e^{\eta-\lambda}\leq e^{\eta-\lambda}\Bigg(\frac{1}{t}-\frac{\Lambda t^2}{C+\frac{\Lambda}{3}t^3}\Bigg)
$$
and, after integration,
\begin{equation}\label{eta-lambda}
e^{\eta-\lambda}\leq 
C\, \frac{t}{C+\frac{\Lambda}{3}t^3}
\leq Ct^{-2}, \qquad t\in[1,T), \ x\in[0,1].
\end{equation}
Next, using (\ref{1.3}), (\ref{e 2eta}), (\ref{e eta+lambda}), and (\ref{eta-lambda}), we have 
\begin{align*}
\int_0^1\eta(t,x)\ dx
&=\int_0^1\etab(x)\ dx+\int_1^t\int_0^1\eta_t(s,x)\ dxds
\\
&\leq C+\int_1^t\frac{1}{2s}\int_0^1\big(1+e^{2\eta}(8\pi s^2\rho-\Lambda s^2)\big)\ dxds
\\
& \leq C+\frac{1}{2}\ln t+4\pi\int_1^t\int_0^1 s e^{\eta-\lambda}e^{\eta+\lambda}\rho\ dxds-\int_1^t\int_0^1  \frac{\Lambda}{2}s e^{2\eta}\ dxds, 
\end{align*}
thus 
\begin{align*}
\int_0^1\eta(t,x)\ dx
&\leq C+\frac{1}{2}\ln t+C\int_1^ts^{-5}\ ds-\frac{1}{2}\int_1^t\frac{\Lambda s^2}{C+\frac{\Lambda}{3}s^3}\ ds\\&\leq C+\frac{1}{2}\ln \Bigg(\frac{\Lambda t}{C+\frac{\Lambda}{3}t^3}\Bigg).
\end{align*}

It then follows from (\ref{eta-int}) that
$$
\eta(t,x)\leq C(1+t^{-3})+\frac{1}{2}\ln \Bigg(\frac{\Lambda t}{C+\frac{\Lambda}{3}t^3}\Bigg),
$$
which leads to an upper bound for $\eta$, i.e. 
$$
e^{2\eta(t,x)}\leq Ct^{-2}, \qquad t\in[1,T), \ x\in[0,1],
$$
and the proof is complete.
\end{proof}


\subsection*{Late-time asymptotics}

We determine now
 the explicit leading asymptotic behavior of $r$, $s$, $\eta$, $\lambda$, $\lambda_t$, $\eta_t$ and $\eta_x$,
and then check that each of the generalized Kasner exponents tends to $1/3$.  
We have proven that (see equation (\ref{bound}))
\begin{equation}\label{r s}
 r=O(t^{-1}), \qquad s=O(t^{-1}).
\end{equation}
and the equation (\ref{eta}) implies 
\[
(t e^{-2\eta})_t=\Lambda t^2-4\pi t^2(r+s).
\]
Integrating over $[1,t]$ and using (\ref{r s}), we obtain 
\[
\Big| t e^{-2\eta}-\frac{\Lambda}{3}t^3 \Big| \leq Ct^2,
\]
that is, 
$e^{-2\eta}=(\Lambda/3)t^2(1+O(t^{-1}))$, 
so that
$$ 
e^{\eta}=\sqrt{\frac{3}{\Lambda}}t^{-1}(1+O(t^{-1})).
$$ 
In view of $\eta_t= (1/2t) - (\Lambda/2) t e^{2\eta}+2\pi t e^{2\eta}(r+s)$, one has 
\begin{equation}\label{eta t decay}
\eta_t=-\frac{1}{t}(1+O(t^{-1})),
\end{equation}
and, after integration over $[1,t]$,
$\eta=-\ln t(1+O\big((\ln t)^{-1})\big)$.  

Since $\lambda_t=\eta_t+\Lambda t e^{2\eta}- (1/t)$, one also has 
\begin{equation}\label{lambda t decay}
\lambda_t=\frac{1}{t}(1+O(t^{-1})),
\end{equation}
and integrating over$[1,t]$ gives
$\lambda=\ln t(1+O\big((\ln t)^{-1})\big)$. 
This implies $e^\lambda=O(t)$, 
and recalling that $\eta_x=-2\pi t e^{\lambda+\eta}(r-s)$ one deduces that
\begin{equation}\label{eta x decay}
 \eta_x=O(1).
\end{equation}

Consider the generalized Kasner exponents which take the following form 
for the metric under consideration
 (see for instance \cite{rein}): 
\be
\label{Kasner-expo}
\frac{\kappa^1_1(t,x)}{\kappa(t,x)}=\frac{t\lambda_t}{t\lambda_t+2}, 
\qquad
 \frac{\kappa^2_2(t,x)}{\kappa(t,x)}=\frac{\kappa^3_3(t,x)}{\kappa(t,x)}=\frac{1}{t\lambda_t+2},
\end{equation}
 where $\kappa(t,x)=\kappa^i_i(t,x)$ is the trace of the second fundamental form $\kappa_{ij}(t,x)$ of the metric. It follows from (\ref{lambda t decay}) that as $t$ tends to $\infty$, each of these quantities tends to $1/3$, uniformly in $x$.


\subsection*{Future geodesic completeness}

The late-time asymptotic expansion above allows us to establish that the spacetime is future geodesically complete, as follows. 
Let $\tau\mapsto \big(\gamma^\alpha\big)(\tau)$ (with $t=\gamma^0(\tau)$)
be a future directed causal geodesic 
defined on an interval $[1,\tau_+)$ with $\tau_+$ maximal, 
and 
normalized so that 
$\gamma^0(\tau_0)=t(\tau_0)=1$ for some $\tau_0\in[1,\tau_+)$. 
We are going to prove that $\tau_+=+\infty$. 

Since $\gamma$ is causal and future directed, we have 
$$
 g_{\alpha\beta}\gamma_\tau^\alpha\gamma_\tau^\beta=-m^2,
\qquad
\gamma_\tau^0>0,
$$
where $m=0$ if $\gamma$ is null, and $m\neq0$ if $\gamma$ is timelike.
Since $\frac{dt}{d\tau}=\gamma^0_\tau>0$, the geodesic can be parametrized by the coordinate time $t$. With respect to 
this coordinate time the geodesic exists on the whole interval $[1,+\infty)$ 
since on each bounded interval of $t$ the Christoffel symbols are bounded and the right-hand sides of the geodesic equation
(written in coordinate time) are linearly bounded in $\gamma^1_\tau$, $\gamma^2_\tau$, $\gamma^3_\tau$.

Along the geodesic we define 
$$
w:=e^\lambda \gamma^1_\tau, 
\qquad
F:=t^4 \, \Big((\gamma^2_\tau)^2+(\gamma^2_\tau)^3\Big). 
$$
Using the geodesic equation it is easily checked that
$$
 \frac{dw}{d\tau}=-\lambda_t\gamma^0_\tau w-e^{2\eta-\lambda}\eta_x(\gamma_\tau^0)^2, 
\qquad  \frac{dF}{d\tau}=0.
$$
The relation between coordinate time and proper time is then given by
\begin{equation}\label{proper}
 \frac{d\tau}{dt}=(\gamma^0_\tau)^{-1}=\frac{e^{\eta}}{\sqrt{m^2+w^2+F/t^2}}.
\end{equation}
We will now exhibit a lower bound for $d\tau/dt$ by a function with divergent integral on $[1, +\infty)$ 
and, to this end, 
an estimate on $w$ as a function of the coordinate time is needed.

Assume that $w(t)>0$ for some $t\geq 1$. Then, as long as $w(s)>0$, we have
\begin{align}\label{d w}
 \frac{dw}{ds}&=-\lambda_t w-e^{\eta-\lambda}\eta_x\sqrt{m^2+w^2+F/s^2}\notag \\
&=4\pi se^{2\eta}(j\sqrt{m^2+w^2+F/s^2}-\rho w)+\frac{1}{2t}w-\frac{\Lambda}{2}se^{2\eta}w.
\end{align}
Using the elementary inequality $\sqrt{a+b}\leq\sqrt{a}+\sqrt{b}$ and the equation (\ref{e 2eta}), we obtain 
\begin{align*}
\frac{dw}{ds}\leq4\pi se^{2\eta}(|j|-\rho)w+\frac{1-\Lambda s^2 e^{2\eta}}{2s}w+4\pi se^{2\eta}|j|\sqrt{m^2+F/s^2}.
\end{align*}
We can drop the first two terms which are negative since
 $|j|\leq\rho$ and $$1-\Lambda s^2 e^{2\eta}\leq\frac{C}{\Lambda}s^{-3}-2<0, 
\qquad
 s \ \rm{ sufficiently large},$$ and we estimate the third term by $C s^{-2}$
 (since $|j|\leq Cs^{-1}$ and $e^{2\eta}\leq Cs^{-2}$). It then follows that
\begin{equation}\label{dw ds}
 \frac{dw}{ds}\leq Cs^{-2}.
\end{equation}

Let $t_0\in[1,t)$ be the smallest time such that $w(s)>0$ for all $s\in[t_0,t)$. Then
integrating (\ref{dw ds}) over $[t_0,t]$ gives $$w(t)\leq C.$$
For the case $w(t)<0$,  it follows from (\ref{d w}) that, as long as $w(s)<0$
\begin{align*}
\frac{dw}{ds}&\geq4\pi se^{2\eta}(-\rho\sqrt{m^2+w^2+F/s^2}-\rho w)+\frac{1-\Lambda s^2e^{2\eta}}{2s}w\\
&\geq-4\pi se^{2\eta}\rho\sqrt{m^2+F/s^2}+8\pi se^{2\eta}\rho w\\
&\geq Cs^{-2}(-1+w),
\end{align*}
we have used the fact that $|j|\leq \rho$, $\frac{1-\Lambda s^2e^{2\eta}}{2s}<0$ for large $s$ and the elementary inequality $\sqrt{a+b}\leq\sqrt{a}+\sqrt{b}$.
Therefore we have 
\begin{equation}\label{dw1 ds}
 \frac{1}{1-w}\frac{d(1-w)}{ds}\leq C s^{-2}.
\end{equation}
Let $t_1\in[1,t)$ be the smallest time such that $w(s)<0$ for all $s\in[t_1,t)$. Then
integrating (\ref{dw1 ds}) over $[t_1,t]$ implies $$-w(t)\leq C.$$ 

In either case, we arrive at
$$
|w(t)|\leq C,  \qquad t\geq 1.
$$
On the other hand equation (\ref{e 2eta}) implies that 
$$
e^\eta\geq C t^{-1}, \qquad t\geq 1,
$$
so we then deduce from (\ref{proper}) that
$$
 \frac{d\tau}{dt}\geq\frac{Ct^{-1}}{\sqrt{m^2+C+F}},
$$
and since the integral of the right-hand side over $[1,+\infty)$ diverges, it follows that $\tau_+=+\infty$
and the proof of future geodesic completeness is completed.


\section*{Acknowledgments}

This work was completed when the first author (PLF) gave a short course at the thirteen 
GIRAGA seminar hold at the University of Yaounde 
in September 2010; he is particularly grateful to D. B\'ekoll\'e and the organizing committee
for their invitation and warm welcome.  PLF was supported by the Centre National de la Recherche Scientific and 
the Agence Nationale de la Recherche (ANR) through Grant 06-2-134423: 
``Mathematical Methods in General Relativity''.   


\end{document}